\documentclass[runningheads]{llncs}
\usepackage{graphicx}
\usepackage[ruled,vlined,linesnumbered]{algorithm2e}
\usepackage{amsmath,amssymb,txfonts,caption}
\usepackage{hyperref}
\begin{document}
\title{Asynchronous Gathering of Robots with Finite Memory on a Circle under Limited Visibility\thanks{The first three authors are full time research scholars in Jadavpur University.}}
\titlerunning{Asynchronous Gathering of Robots on a Circle}
%
\author{Satakshi Ghosh\orcidID{0000-0003-1747-4037}\and
Avisek Sharma\orcidID{0000-0001-8940-392X}\and
Pritam Goswami\orcidID{0000-0002-0546-3894}\and
Buddhadeb Sau\orcidID{0000-0001-7008-6135}}
\authorrunning{S. Ghosh, A. Sharma, P. Goswami, and B. Sau}
%
\institute{Jadavpur University, 188, Raja S.C. Mallick Rd,
Kolkata 700032, India\\
\email{\{satakshighosh.math.rs,aviseks.math.rs, pritamgoswami.math.rs,buddhadeb.sau\}@jadavpuruniversity.in}}
\maketitle             
\begin{abstract}
Consider a set of $n$ mobile entities, called robots, located and operating on a continuous circle, i.e., all robots are initially in distinct locations on a circle. The \textit{gathering} problem asks to design a distributed algorithm that allows the robots to assemble at a point on the circle. Robots are anonymous, identical, and homogeneous. Robots operate in a deterministic Look-Compute-Move cycle within the circular path. Robots agree on the clockwise direction. The robot's movement is rigid and they have limited visibility $\pi$, i.e., each robot can only see the points of the circle which is at an angular distance strictly less than $\pi$ from the robot.

Di Luna \textit{et al}. [DISC'2020] provided a deterministic gathering algorithm of oblivious and silent robots on a circle in semi-synchronous (\textsc{SSync}) scheduler. Buchin \textit{et al}. [IPDPS(W)'2021] showed that, under full visibility, $\mathcal{OBLOT}$ robot model with \textsc{SSync} scheduler is incomparable to $\mathcal{FSTA}$ robot (robots are silent but have finite persistent memory) model with asynchronous (\textsc{ASync}) scheduler. Under limited visibility, this comparison is still unanswered. So, this work extends the work of Di Luna \textit{et al}. [DISC'2020] under \textsc{ASync} scheduler for $\mathcal{FSTA}$ robot model. 

\end{abstract}

\section{Introduction}
In swarm robotics, robots achieving some tasks with minimum capabilities is the main focus of interest. In the last two decades, there are huge research interest in robots working on coordination problems. It is not always easy to use robots with strong capabilities in real-life applications, as making these robots is not at all cost-effective. If a swarm of robots with minimum capabilities can do the same task then it is effective to use swarm robots rather than using robots with many capabilities, as making the cost of these robots in the swarm is very much cheaper than making robots with many capabilities. Also, it is very easy to design a robot of a swarm because they have minimum capabilities. 

The gathering is a very vastly studied problem by researchers. This is one of the fundamental tasks of autonomous robots in the distributed system. Gathering means if there are $n$ number of robots initially positioned arbitrarily then within finite time all robots will meet at a point that is not fixed initially. When there are two robots then this task is called $\mathit{rendezvous}$. It is not always easy to meet at a point by very weak robots in the distributed system. So it is challenging to design a distributed algorithm to gather some robots with very inexpensive and with few capabilities. 

In the theoretical approach, robots are assumed to be autonomous, identical, and homogeneous. They have no global coordinate system. Depending on the capabilities there are generally four types of robot models. These models are $\mathcal{OBLOT}$, $\mathcal{FSTA}$, $\mathcal{FCOM}$ and $\mathcal{LUMI}$. In each of these models, robots are assumed to be autonomous (i.e the robots do not have any central control), identical (i.e the robots are physically indistinguishable), and homogeneous (i.e each robot runs the same algorithm). Furthermore in the $\mathcal{OBLOT}$ model, the robots are silent (i.e there is no means of communication between the robots) and oblivious (i.e the robots do not have any persistent memory to remember their previous state), in $\mathcal{FSTA}$ model the robots are silent but not oblivious, in $\mathcal{FCOM}$ model the robots are oblivious but not silent and in $\mathcal{LUMI}$ model robots are neither silent nor oblivious. The robots after getting activated operate in a \textsc{Look-Compute-Move} (LCM) cycle. In the \textsc{Look} phase a robot takes input from its surroundings and then with that input runs the algorithm in \textsc{Compute} phase to get a destination point as an output. The robot then goes to that destination point by moving in the \textsc{Move} phase. A scheduler controls the activation of the robots. There are mainly three types of schedulers considered in the literature. In a synchronous scheduler, time is divided into global rounds. In (\textsc{FSync}) scheduler each robot is activated in all rounds and executes the LCM cycle simultaneously. In a semi-synchronous scheduler (\textsc{SSync})  all robots may not get activated in each round. But the robots that are activated in the same round execute the LCM cycle simultaneously. Lastly in the asynchronous scheduler (\textsc{ASync}), there is no common notion of time, a robot can be activated at any time. There is no concept of rounds. So there is no assumption regarding synchronization.

In this work, we investigate the gathering of robots on a circle. The robots are on a perimeter of a circle and they can only move along the perimeter of that circle. Here the robots have limited visibility. If the robots have full visibility and if they can elect a unique leader then gathering of robots is not so hard even if in \textsc{ASync} scheduler. But with limited visibility, as unique leader selection is not easy, so the gathering is also not trivial. Here we investigate the gathering of robots in an asynchronous scheduler on a circle with limited visibility with $\mathcal{FSTA}$ robots.

\section{Related work}
In distributed computing gathering of mobile robots has been the focus of intensive investigation under various computational power and communication varieties. The gathering has been extensively researched in distributed computing. This problem has been thoroughly studied both in continuous and discrete domains. There are few surveys on this problem \cite{BhagatMM19,flochhini19,CieliebakFPS12}. In the continuous case, both the gathering and the rendezvous problems have been investigated in the context of swarms of autonomous mobile robots operating in one and two-dimensional spaces, requiring them to meet at or converge to the same point. Here we mainly focus on the gathering of robots when a robot has limited visibility. Flochhini \cite{Flocchini05} in their paper showed that gathering is done by robots in a plane under limited visibility in \textsc{ASync} scheduler, assuming that all robots agree on the orientation of both coordinate axes. In \cite{PoudelS17} authors provided an optimal algorithm for the gathering of robots on a plane under limited visibility. In another work on \cite{Souissi06} authors solved the gathering in \textsc{Ssync} scheduler with one axis agreement. In \cite{ando99} authors solved that without chirality and any agreement on local coordinate system robots gather (similarly converge) under \textsc{Fsync} (\textsc{SSync}) scheduler. After that in paper \cite{degener11} they showed that the time complexity depends on the number of swarms.

In the discrete case, the robots are dispersed in a network modeled as a graph and are required to gather at the same node and terminate. The main difficulty in solving the problem is symmetry. But when the network nodes are anonymous, the network is symmetric, the mobile agents are identical, and there is no means of communication, it is impossible to solve the problem in deterministic means. So the focus is to make the problems deterministically solvable with minimal assumptions. In the paper \cite{GuilbaultP11}, authors showed the gathering of asynchronous oblivious robots in a bipartite graph with a local vision which means a robot can see its immediate neighbor only. Also in \cite{'CzyzowiczKP12'} they provide an algorithm of \textit{rendezvous} problem in an arbitrary graph. Another interesting way of researching distributed systems is to investigate any distributed problem with faulty robots. In \cite{BouchardDD15} authors studied the gathering of robots in a network with byzantine robots. Also in \cite{'agmon-06',DefagoP0MPP16} authors provide a gathering algorithm with the crash and byzantine faults. In recent days researchers started to investigate dynamic graphs {\cite{Casteigts10}, that is graphs where the topological changes are not localized, the topology changes continuously and at unpredictable locations, and these changes are an integral part of the nature of the system. In paper \cite{LunaFPPSV17} they start the investigation of gathering in dynamic rings. In \cite{LunaFSVY20a,jczyzowicz2013} they studied the gathering of robots in polygon terrain. In \cite{GoswamiSGS22} they provide an optimal gathering of robots on a triangular grid. There is huge research work on the gathering of robots on grid networks \cite{DAngeloSKN16,BhagatCDM22}.

In our paper robots are on a continuous circle, and we will solve all robots gather at a point on the circle. There are some recent works on this model. In \cite{HuusK15,FeinermanKKR14} authors showed the \textit{rendezvous} of mobile agents with different speeds in a cycle. Since leader election (\cite{Dieudonn-09}) is an important task for solving many problems in distributed computing. In \cite{'FlocchiniKKSY19'} authors start the investigation of solving \textsc{Gather} and \textsc{elect} by the set of robots $R$ deployed in a continuous cycle $\mathcal{C}$, they primarily focus on \textsc{Elect}. But in their model robots have no visibility i.e they can not see at any distance. After that Di Luna \textit{et al.} \cite{'LunaUVY20'}, proves that with limited visibility $\pi$ gathering of anonymous and oblivious robots is possible in \textsc{SSync} scheduler but it is impossible to gather all robots when visibility is $\frac{\pi}{2}$ on continuous circle $\mathcal{C}$. In this paper, we proposed an algorithm for gathering with $n$ robots with finite time under \textsc{ASync} schedulers when visibility is $\pi$ on a circle in $\mathcal{FSTA}$ robot model. 

\section{Our contribution}
Under full visibility, gathering on a circle by robots in an asynchronous scheduler is not so difficult when robots agree on a global sense of handedness. But excluding one point from visibility makes the gathering problem a bit difficult. Di Luna \textit{et al.} (\cite{'LunaUVY20'}) showed the gathering of oblivious robots on a circle under visibility $\pi$ in \textsc{SSync} scheduler. In this paper, we extend this work under the asynchronous scheduler. Here the robots are anonymous, identical, and silent. However, we only compromise the obliviousness of the robots to get the liberty of asynchronous scheduler. We assume each robot has finite persistent memory. 

In \cite{BFKSW21}, authors showed that under full visibility, semi-synchronous scheduler with oblivious, silent robots, abbreviated by $\mathcal{OBLOT}^S$ is incomparable with an asynchronous scheduler with silent finite memory robots, abbreviated as $\mathcal{FSTA}^A$. For the limited visibility model, this comparison is still not done. This result shows that bringing finite memory into the picture does not make the model equally powerful as $\mathcal{OBLOT}^S$. Having finite memory is not at all a big deal in comparison to maintaining a synchronous scheduler.

Further, however, we do not consider a general asynchronous scheduler. Precisely, we assumed that in the LCM cycle of a robot, it takes a nonzero time (non instantaneous) to finish its look and compute phase together. This assumption is fairly realistic in a practical scenario.  

\section{Model and Definitions}
\subsection{Robot Model}

 In the problem, we are considering $\mathcal{FSTA}$ robot model. The robots are anonymous, and identical, but not oblivious. Robots have finite persistent memory. Robots can not communicate with each other. Robots have weak multiplicity detection capability, i.e., robots can detect a multiplicity point, but can not determine the number of robots present at a multiplicity point. All robots are placed on a circle. The robots agree on a global sense of handedness. Robots operate in \textsc{Look-Compute-Move} cycle. In each cycle a robot takes a snapshot of the positions of the other robots according to its own local coordinate system (\textsc{Look}); based on this snapshot, it executes a deterministic algorithm to determine whether to stay put or to move to another point on the circle. (\textsc{Compute}); and based on the algorithm the robots either remain stationary or makes a move to an adjacent point (\textsc{Move}). We assume that robots are controlled by a fully asynchronous adversarial scheduler (\textsc{ASync}). The robots are activated independently and each robot executes its cycles independently. This implies the amount of time spent in \textsc{Look}, \textsc{Compute}, \textsc{Move}, and inactive states are finite but unbounded, unpredictable, and not the same for different robots. We assume that look and compute phase together in an LCM cycle of a robot is not instantaneous. The robots have no common notion of time. All robots move at the same speed and their movement is rigid. Here the initial configuration is asymmetric. Robots have limited visibility which means a robot can not see the entire circle. Let $a$ and $b$ be two points on circle $\mathcal{C}$, then the angular distance between $a$ and $b$ is the measure of the angle subtended at the center of $\mathcal{C}$ by the shorter arc with endpoints $a$ and $b$. A robot has visibility $\pi$, that it can see all the other robots which are with angular distance less than $\pi$. 
 \subsection{Definitions and preliminaries}
\begin{definition}[Configuration]
Let there be a circle given and a finite number of robots be placed on a circle. We call it a configuration of robots on the circle, or simply a configuration.
\end{definition}

\begin{definition}[Rotationally Symmetric Configuration]
For a configuration with no multiplicity point is said to be rotationally symmetric if there is a nontrivial rotation with respect to the center which leaves the configuration unchanged.
\end{definition}
 
\begin{definition}[Antipodal robot]
A robot $r$ is said to be an antipodal robot if there exists a robot $r'$ on the angular distance $\pi$ of the robot. In such a case, $r$ and $r'$ are said to be antipodal robots to each other.
\end{definition}
Note that a robot which is not antipodal is said to be non antipodal robot.
\begin{definition}
Let $r'$ and $r''$ be two robots in a configuration positioned at distinct positions. Then $cwAngle(r',r'')$ ($ccwAngle(r',r'')$) is the angular distance from $r'$ to $r''$ in clockwise (counter clockwise) direction.
\end{definition}

\begin{definition}[Angle sequence]
Let $r$ be a robot in a given configuration with no multiplicity point and let $r_1, r_2,\dots, r_n$ be the other robots on the circle in clockwise order. Then the angular sequence for robot $r$ is the sequence $$(cwAngle(r,r_1),cwAngle(r_1,r_2),cwAngle(r_2,r_3),\dots,cwAngle(r_n,r)).$$ We denote this sequence as $\mathcal{S}(r)$. We further denote the sub sequence $$(cwAngle(r,r_1),cwAngle(r_1,r_2),cwAngle(r_2,r_3),\dots,cwAngle(r_{i-1},r_i))$$ of $\mathcal{S}(r)$ as $\mathcal{S}(r,r_i)$. Further, we call $cwAngle(r,r_1)$ as the leading angle of $r$.

\end{definition}
Note that as the configuration is initially rotationally asymmetric, by results from the paper \cite{'LunaUVY20'} we can say that all the robots have distinct angle sequences.
\begin{definition}[Lexicographic Ordering]
Let $\tilde{a} = (a_1,\dots,a_n)$ and $\tilde{b}=$ $(b_1,\dots,b_n)$ be two finite sequences of reals of same length. Then $\tilde{a}$ is said to be lexicographically strictly smaller sequence if $a_1<b_1$ or there exists $1<k< n$ such that $a_i=b_i$ for all $i=1,2,\dots,k$ and $a_{k+1}<b_{k+1}$. $\tilde{a}$ is said to be lexicographically smaller sequence if either $\tilde{a}=\tilde{b}$ or $\tilde{a}$ is lexicographically strictly smaller sequence.
\end{definition}
\begin{definition}[True leader]
In a configuration, a robot with lexicographically smallest angular sequence is called a true leader.
\end{definition}
If the configuration is rotationally asymmetric and contains no multiplicity point, there exists exactly one robot which has strictly the smallest lexicographic angle sequence. Hence there is only one true leader for such a configuration.
Since a robot on the circle cannot see whether its antipodal position is occupied by a robot or not. So a robot can assume two things: 1)~the antipodal position is empty, let's call this configuration $C_0(r)$ 2)~the antipodal position is nonempty, let's call this configuration $C_1(r)$. So a robot $r$ can form two angular sequences. One considering $C_0(r)$ configuration and another considering $C_1(r)$. The next two definitions are from the viewpoint of a robot. If the true leader robot can confirm itself as the true leader, we call it \textit{Sure Leader}. If the true leader or some other robot has an ambiguity of being a true leader depending on the possibility of $C_0(r)$ or $C_1(r)$ configuration, then we call it $\textit{Confused Leader}$. There may be the following possibilities. 
\begin{itemize}
    \item Possibility-1: $C_0(r)$ configuration has rotational symmetry, so $C_1(r)$ is the only possible configuration.
    \item Possibility-2: $C_1(r)$ configuration has rotational symmetry, so $C_0(r)$ is the only possible configuration.
    \item Possibility-3: Both $C_0(r)$ and $C_1(r)$ has no rotational symmetry, so both $C_0(r)$ and $C_1(r)$ can be possible configurations.
\end{itemize}
\begin{definition}[Sure leader]
A robot $r$ in a rotationally asymmetric configuration with no multiplicity point is called Sure Leader if $r$ is the true leader in any possible configuration.
\end{definition}

Note that, the Sure leader is definitely the true leader of the configuration. Hence at any time if the configuration is asymmetric and contains no multiplicity point, there is at most one Sure leader. 


\begin{definition}[Confused leader]
A robot $r$ in a rotationally asymmetric configuration with no multiplicity point is called a Confused Leader if both $C_0(r)$ and $C_1(r)$ are possible configurations and $r$ is a true leader in one configuration but not in another.
\end{definition}
\begin{definition}[Follower robot]
A robot in an asymmetric configuration with no multiplicity point is said to be a follower robot if it is neither a sure leader nor a confused leader.
\end{definition}
\begin{definition}[Expected leader]
A robot in an asymmetric configuration with no multiplicity point is said to be an expected leader if it is not a follower robot. That is, an expected leader is either a sure leader or a confused leader.
\end{definition}
Note that the above definitions are set in such a way that the sure leader (or, a confused leader or, a follower robot) can recognize itself as a sure leader (or, a confused leader or, a follower robot).  
\begin{definition}
For two robots $r$ and $r'$ situated at different positions on the circle such that $cwAngle(r,r')=\theta$ we define $[r,r']$ as the set of points $x$ on the circle such that  $0\le cwAngle(r,x) \le \theta$ and $(r,r')$ as the set of points $x$ on the circle such that  $0 < cwAngle(r,x) < \theta$.
\end{definition}
\begin{definition}
Let $r$ be a robot in a configuration, then another robot $r_1$ is said to be situated at the left of $r$ if $cwAngle(r,r_1)>\pi$ and said to be at right if $cwAngle(r,r_1)<\pi$.
\end{definition}
Note that, the true leader of the configuration can become a confused leader and a robot other than the true leader can also become a confused Leader. The next results lead us to find the maximum possible number of expected leaders.
\begin{proposition}\label{lemma1}
Let $(a_1,\dots,a_k)$ be the angular sequence of the true leader, say $r_0$, in a rotationally asymmetric configuration with no multiplicity point. Then there cannot be another robot $r'$ with the following properties.
\begin{enumerate}
    \item $r'$ is at left side to $r_0$,
    \item $\mathcal{S}(r',r_0)=(a_1,a_2,\dots,a_i)$,
    \item First $i$ angles of $\mathcal{S}(r_0)$ respectively are $a_1,a_2,\dots,a_{i-1}$ and $a_i$. 
\end{enumerate}
\end{proposition}
\begin{proof}
We prove this result by contradiction. If possible let there be such a robot $r'$. Then note that $\mathcal{S}(r_0)$ is $\tilde{a}=(a_1,\dots,a_i,a_{i+1},\dots,a_{k-i},a_1,\dots,a_i)$ and $\mathcal{S}(r')$ is $\tilde{a}_1=(a_1,\dots,a_i,a_1,\dots,a_i,a_{i+1},\dots,a_{k-i})$. Since $r_0$ has the strictly smallest angular sequence, so $a_1$ is the smallest angle in the configuration and also $a_{i+1}\le a_1$, which leads to $a_{i+1}=a_1$. Next, we show that $a_2=a_{i+2}$. Since the $\tilde{a}$ is strictly the smallest angular sequence so $a_{i+2}\le a_2$. If $a_{i+2}<a_2$, then the angular sequence $(a_{i+1},\dots,a_{k-i},a_1,\dots,a_i,a_1,\dots,a_i)$ is smaller than $\tilde{a}$, which is a contradiction. Hence $a_2=a_{i+2}$. Therefore by a similar argument, we can show that $a_{i+j}=a_j$ for $j=3,\dots,i$. Now if $2i= k$ then we see that $\tilde{a}=\tilde{a}_1$, which contradicts the fact that the configuration is rotationally asymmetric. Otherwise, proceeding similarly we can show that $a_{2i+j}=a_j$, for $j=1,2,\dots,i$. Repeating the same argument we can show that $a_{pi+j}=a_j$, for $j=1,2,\dots,i$. Since there are finitely many angles, so after finite number of steps we must end up having that $k=ti$ where $t\ge 2$ and $\tilde{a}=\tilde{a}_1=(a_1,\dots,a_i,a_1,\dots,a_i,\dots,a_1,\dots,a_i)$, which is again a contradiction.
 \end{proof}

Next, we state a simple observation in the following Proposition~\ref{lemma00}.
\begin{proposition}\label{lemma00}
Suppose there is a rotationally asymmetric configuration with no multiplicity point with the true leader, $r_0$ (say), then on including a robot, say $r$, on the circle at an empty point without bringing any rotational symmetry, the true leader of the new configuration must be in $[r_0,r]$.    
\end{proposition}

For a confused leader $r$, there may be two possibilities. The first one is when $r$ is a true leader in $C_0(r)$ configuration but not in $C_1(r)$. The second one is when $r$ is a true leader in $C_1(r)$ configuration but not in $C_0(r)$. We show that the second possibility can not occur. We formally state the result in the following Proposition.

\begin{proposition}\label{lemma30}
If a robot $r$ is a confused leader in asymmetric configuration with no multiplicity point, $r$ is the true leader in $C_0(r)$ configuration but $r$ is not the true leader in $C_1(r)$ configuration.
\end{proposition}
From Proposition~\ref{lemma30} one can observe that if the true leader of an asymmetric configuration with no multiplicity point is a confused leader then its antipodal position must be empty. And also if a robot, which is not the true leader of the configuration, becomes a confused leader then its antipodal position must be non-empty. We record these observations in the following Corollaries.

\begin{corollary}\label{Cor1}
In a rotationally asymmetric configuration with no multiplicity point if the true leader of the configuration is a confused leader then its antipodal position must be empty.
\end{corollary}

\begin{corollary}\label{Cor2}
In a rotationally asymmetric configuration with no multiplicity point if a robot other than the true leader becomes a confused leader then its antipodal position must be non-empty.
\end{corollary}

\begin{proposition}\label{lemma44}
Let $\mathcal{C}$ be a rotationally asymmetric configuration with no multiplicity point and $L$ be the true leader of the configuration, then another expected leader $r$ of $\mathcal{C}$ must satisfy $cwAngle(L,r)\ge\pi$.
\end{proposition}
\begin{proof}
If possible let there be another expected leader $r$ such that $cwAngle(L,r)<\pi$. Since $r$ is not the leader of the configuration, so $r$ must be a confused leader. Hence from Proposition~\ref{lemma30}
we have that $r$ is true leader in $C_0(r)$ configuration and $L$ is the true leader of $C_1(r)$ configuration. This contradicts the Proposition~\ref{lemma00}.

 \end{proof}

\begin{proposition}\label{lemma3}
For any given rotationally asymmetric configuration with no multiplicity point, there can be at most one confused leader other than the true leader.
\end{proposition}
\begin{proof}
Let $\mathcal{C}$ be the given configuration. If possible let $r_1$ and $r_2$ be two confused leaders other than the true leader, say $L$. From Proposition~\ref{lemma44} First we have that $cwAngle(L,r_i)\ge\pi$, for $i=1,2$. Without loss of generality we assume $cwAngle(L,r_1)<cwAngle(L,r_2)$. Since $r_i$s are confused leaders so from Proposition~\ref{lemma30} their antipodal position is non-empty. Let for each $i$, $r_i'$ be the antipodal robot of $r_i$. There are two exhaustive cases. First one is $cwAngle(L,r_1)=\pi$ and second one is $cwAngle(L,r_1)>\pi$.

\textbf{Case-I}: $cwAngle(L,r_1)=\pi$

\textit{Case-IA:} In this case let the first angle of the angle sequences of $r_1$ and $L$ be different. Let the first angle in $\mathcal{S}(L)$ and $\mathcal{S}(r_1)$ are $\theta$ and $\theta_1$ respectively, then $\theta<\theta_1$. Since $r_2'$ cannot be the clockwise neighbor of $L$, so $r_2$ can see the leading angle of $L$. Therefore in order to become a confused leader, the leading angle of $r_2$ has to be $\theta$. Now there may be two cases. Firstly either the clockwise neighbor of $r_2$ is $L$ or not. If the clockwise neighbor of $r_2$ is $L$ then from Proposition~\ref{lemma00} $L$ does not remain the true leader of the configuration, which is a contradiction. Secondly, if the clockwise neighbor of $r_2$ is not $L$ then $r_1$ can see the leading angle of $r_2$, that is $\theta$ which is smaller than the leading angle of $r_1$. This gives $\mathcal{C}\smallsetminus L$ cannot have $r_1$ as a true leader. This contradicts the fact that $r_1$ is a confused leader.

\textit{Case-IB:} Let the first $t$ angles of $\mathcal{S}(L)$ and $\mathcal{S}(r_1)$ are same and $(t+1)^{th}$ angles are different. Let first $t$ clockwise neighbors of $L$ in clockwise order are $x_1,x_2,\dots,x_t$ and first $t$ clockwise neighbors of $r_1$ in clockwise order are $x_1',x_2',\dots,x_t'$. First, we show that $r_2$ is none of $x_i'$s. If possible let $r_2=x_i'$ where $i<t$. Then $r_2$ can see the first $i-1$ angles of $\mathcal{S}(L)$ and $\mathcal{S}(r_1)$. Since $i^{th}$ clockwise neighbor of $r_2$ is at the left of $r_2'$, $r_2$ can see its first $i$ angles of $\mathcal{S}(r_2)$ in original configuration. Now first we observe that the first $i-1$ angles of $r_2$ and $L$ are the same. If not then either $r_2$ would not be a confused leader or $L$ would not be leader of the configuration. Now we see that $i^{th}$ angle of $\mathcal{S}(L)$ and $\mathcal{S}(r_2)$ is also same. If possible let $i^{th}$ angle of $\mathcal{S}(L)$ and $\mathcal{S}(r_2)$ are $\theta_i$ and $\alpha_i$ respectively, and $\alpha_i>\theta_i$ (Note that the case $\alpha_i<\theta_i$ gets excluded from the fact that $L$ is the true leader of the configuration). Since $i^{th}$ angle of $r_1$ is also $\theta_i$ which is visible by $r_2$, so then $r_2$ would not be confused leader. Hence first $i$ angles of $\mathcal{S}(r_1)$ and $\mathcal{S}(r_2)$ are same. Therefore from Proposition~\ref{lemma00}, the $r_2$ cannot be a confused leader. Hence for each $i<t$, $r_2\ne x_i'$. Now it is easy to observe that $r_2\ne x_t$, because otherwise, the leading angle of $r_2$ is strictly greater than the leading angle of $L$. which makes that $r_2$ is not a confused leader.

Hence for each $1\le i\le t$, $r_2\ne x_i'$. Thus note that $r_2$ can see the first $t+1$ angles of $\mathcal{S}(L)$ and in order to remain a confused leader first $t+1$ angles of $\mathcal{S}(r_2)$ should match with it. Now there are two cases, either $(t+1)^{th}$ clockwise neighbor of $r_2$ is $L$ or not. If not then $r_1$ can see the first $t+1$ angles of $\mathcal{S}(r_2)$. And $t+1^{th}$ angle of $\mathcal{S}(r_2)$ is smaller than same of $\mathcal{S}(r_1)$. Therefore $r_1$ would not be a confused leader. In other case if $(t+1)^{th}$ clockwise neighbor of $r_2$ is $L$ then from Proposition~\ref{lemma00}, $L$ does not remain leader of the configuration. Hence for case-I, we end up having a contradiction if there is more than one confused leader other than a true leader.

\textbf{Case-II}: $cwAngle(L,r_1)>\pi$

\textit{Case-IIA:} In this case let the first angle of the angle sequences of $r_1$ and $L$ be different. Let the first angle in $\mathcal{S}(L)$ and $\mathcal{S}(r_1)$ are $\theta$ and $\theta_1$ respectively, then $\theta<\theta_1$. Since in this case $r_2'$ cannot be the clockwise neighbor of $L$, so $r_2$ can see the leading angle of $L$. Therefore in order to become a confused leader, the leading angle of $r_2$ has to be $\theta$. Now $r_1$ can see the leading angle of $r_2$, that is $\theta$ which is smaller than the leading angle of $r_1$. This gives $\mathcal{C}\smallsetminus r_1'$ cannot have $r_1$ as true leader. This contradicts the fact that $r_1$ is confused leader.

\textit{Case-IIB:} Let the first $t$ angles of $\mathcal{S}(L)$ and $\mathcal{S}(r_1)$ are same and $(t+1)^{th}$ angles are different. Let first $t$ clockwise neighbors of $L$ in clockwise order are $x_1,x_2,\dots,x_t$ and $t$ clockwise neighbors of $r_1$ in clockwise order are $x_1',x_2',\dots,x_t'$. Since $r_1$ is not antipodal of $L$, so $r_2'$ cannot be the clockwise first neighbor of $L$. Borrowing the argument from Case-I we can conclude that $r_2$ is none of $x_i'$s. In other cases, $r_2$ can see the first $t+1$ angles of $\mathcal{S}(L)$ and in order to be a confused leader first $t+1$ angles of $\mathcal{S}(r_2)$ should coincide with the same with $\mathcal{S}(L)$. Now we see that $r_1$ can see the first $t+1$ angles of $\mathcal{S}(r_2)$. If not then, since $r_1'$ is at right to $L$, so $L$ will be at most $t^{th}$ neighbor of $r_1$. This implies $L$ is not the true leader of the configuration from Proposition~\ref{lemma00}. So $r_1$ can see the first $t+1$ angles of $\mathcal{S}(r_2)$. Since $(t+1)^{th}$ angle of $\mathcal{S}(r_2)$ is smaller than the same of $\mathcal{S}(r_1)$ then $r_1$ does not remain a confused leader of the configuration. Which is a contradiction. 

Hence it is proved that there cannot be more than one confused leader other than the true leader of the configuration.

 \end{proof}

Let $\mathcal{C}$ be a rotationally asymmetric configuration with no multiplicity point. Then $\mathcal{C}$ will have a true leader and from above Proposition~\ref{lemma00} there can be at most one more confused leader. Hence we can have the following four exhaustive cases for $\mathcal{C}$.
\begin{enumerate}
    \item $\mathcal{C}$ has exactly one expected leader and that is a sure leader.
    \item $\mathcal{C}$ has exactly one expected leader and that is a confused leader.
    \item $\mathcal{C}$ has exactly two expected leaders and both are confused leaders.
    \item $\mathcal{C}$ has exactly two expected leaders. One of them is a sure leader and another one is a confused leader.
\end{enumerate}
The Fig.~\ref{leader0} and Fig.~\ref{leader1} give the existence of all four above cases. For the third case, we observe two properties in the following Propositions.
\begin{figure}[ht] 
    \centering 
     \includegraphics[width=0.6\linewidth]{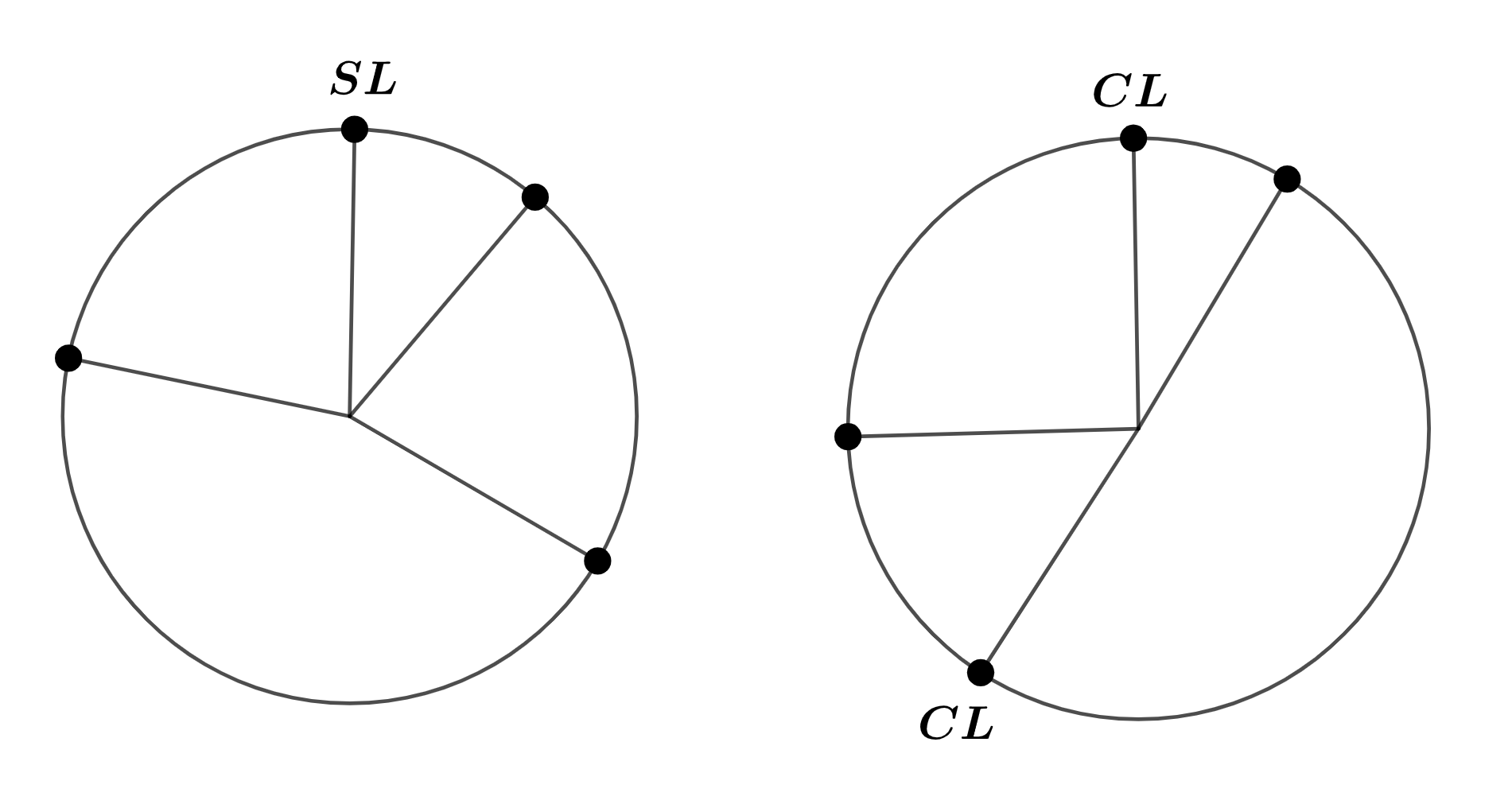}
     \caption{In the left figure only one sure leader (SL) and in right figure two confused leader (CL) in the configuration.}
     \label{leader0}
    \end{figure}
\begin{figure}[ht] 
    \centering 
     \includegraphics[width=0.6\linewidth]{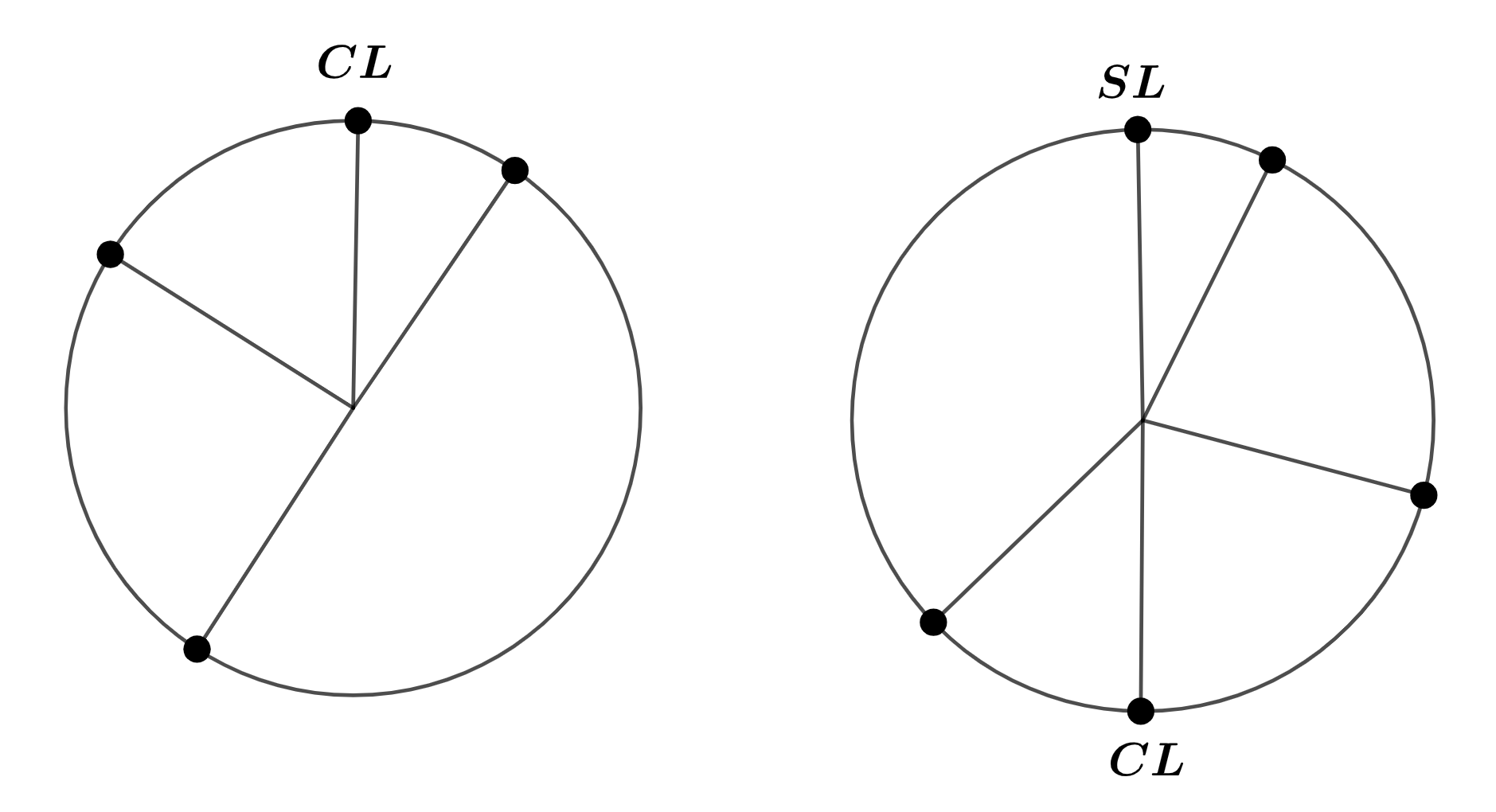}
     \caption{In left figure only one confused leader (CL) and in the right figure one sure leader (SL) and one confused leader (CL).}
     \label{leader1}
    \end{figure}

\begin{proposition}\label{lemma2}
If for a rotationally asymmetric configuration with no multiplicity point if there are two confused leaders then they can not be antipodal of each other.
\end{proposition}
\begin{proof}
Let $p$ and $q$ be two confused leaders of a rotational asymmetric configuration with no multiplicity point. Now one of $p$ and $q$ must be the true leader of the configuration. Without loss of generality let $p$ be the leader of the configuration. Then from the Corollary~\ref{Cor1} antipodal position of $p$ must be empty. Hence another confused leader $q$ cannot be at the antipodal position of $p$.
 \end{proof}

 \begin{proposition}\label{lemma3oo}
If for a rotationally asymmetric configuration with no multiplicity point if there are two confused leaders then their clockwise neighbors cannot be antipodal to each other.
 \end{proposition}
 \begin{proof}
If possible let the clockwise neighbors of two confused leaders be antipodal to each other. Let $L$ be the true leader of the configuration and $r$ be another confused leader. Then from Proposition~\ref{lemma44} and Proposition~\ref{lemma2} we have $cwAngle(L,r)>\pi$. Then the leading angle of $L$ becomes strictly greater than the leading angle of $r$ (Figure~\ref{Fig:3}). This contradicts the fact that $L$ is the true leader of the configuration. 
 \end{proof}
 \begin{figure}[ht] 
     \centering      \includegraphics[width=0.3\linewidth]{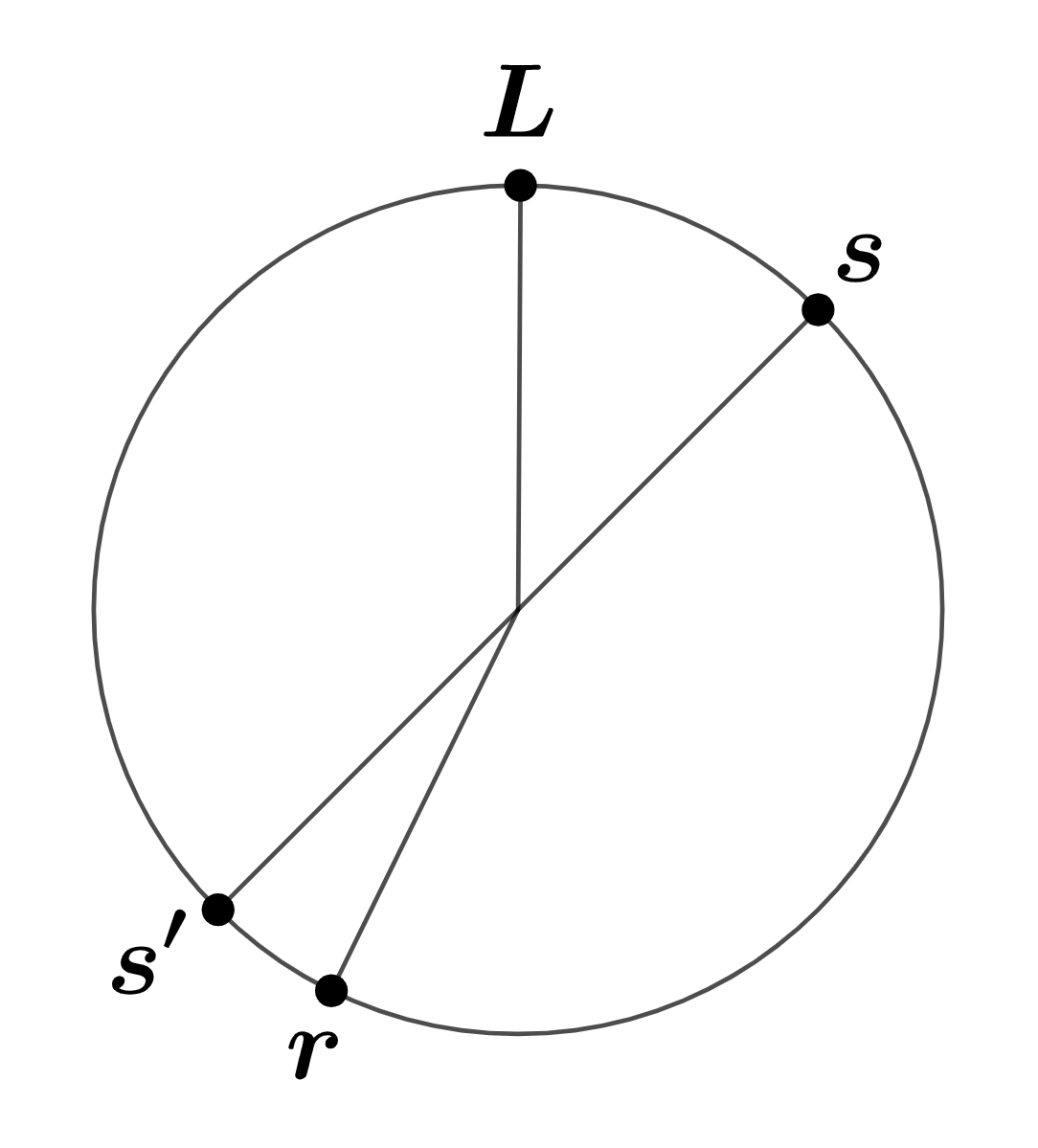}
     \caption{An image related to Proposition~\ref{lemma3oo}}
     \label{Fig:3}
    \end{figure}

  \section{Proposed Algorithm}
In this section, we provide the proposed algorithm that works for asynchronous robots with finite memory. The initial configuration is rotationally asymmetric and the robots are at distinct points on the circle. Each robot has $\pi$ visibility, which means, a robot cannot see its antipodal position. Each robot can move at a uniform speed on the arc of the circle. The moving speed of each circle is same. The proposed algorithm is provided in the Algorithm~\ref{algo:gathering}.

\paragraph*{Discussion of the Algorithm:} 
 Robots are initially dispersed on distinct points of a circle and the initial configuration is rotationally asymmetric. Robots have limited visibility. A robot can not see its antipodal position.
 As the initial configuration is rotationally asymmetric so each robot has distinct angle sequences. On getting activated a robot $r$ first checks whether it is the only robot in its visibility. If yes, then it moves $\pi/2$ angle clockwise. Otherwise, it looks for a multiplicity point. If there is any multiplicity point in the visible configuration then there are two possibilities. Possibility number one, $r$ itself is located at a multiplicity point, and possibility number two, is not located at a multiplicity point. If $r$ is not located at a multiplicity point, and one of the multiplicity points is a clockwise or anticlockwise neighbor of $r$, then $r$ moves to the closer multiplicity point. Suppose $r$ is located on a multiplicity point. For such a case, if there is another multiplicity point other than where $r$ is located, and that multiplicity point is at a clockwise distance less than $\pi$ then $r$ moves to that multiplicity point. Next, if there is no multiplicity point in the visible configuration, then it decides whether it is an expected leader or a follower robot. If $r$ is a follower robot then, it does nothing. If $r$ is a sure leader and its state is \texttt{off}, then it moves to its first clockwise neighbor. If $r$ is a confused leader then it checks its state. If its state is \texttt{terminate} then it does nothing. If its state is \texttt{off} then there are several cases.
 Case-I: If the first clockwise neighbor of $r$ is safe, then $r$ moves to its neighbor's position.
 Case-II: If the first clockwise neighbor of $r$ is not safe and $C_0(r)$ configuration has another confused leader other than $r$, then $r$ does nothing.
 Case-III: If neither of Case-I or Case-II holds then $r$ changes its state to \texttt{moveHalf} and moves $\theta/2$ angular distance clockwise.
 
 Next, suppose on getting activated, $r$ is at the state \texttt{moveHalf}. If the first clockwise neighbor of $r$, say $r'$, is not antipodal then $r$ changes its state to \texttt{terminate} and does not move. Otherwise, let $s$ be the robot at the antipodal position of $r'$ and the leading angle of $r$ is $\theta/2$. Then if there is a robot visible to $r$ in the arc $[s-\frac{\theta}{2}R,s+\frac{\theta}{2}R)$ then $r$ changes its state to terminate and move $\theta/2$ angular distance counterclockwise. If $r'$ is antipodal and there is no robot visible to $r$ in the arc $[s-\frac{\theta}{2}R,s+\frac{\theta}{2}R)$, then $r$ changes the state to \texttt{moveMore} and moves $\theta/4$ angular distance clockwise.
 
 Next, suppose on getting activated, $r$ is at the state \texttt{moveMore}. If the first clockwise neighbor of $r$, say $r'$, is not antipodal then $r$ changes its state to \texttt{terminate} and does not move. Otherwise, let $s$ be the robot at the antipodal position of $r'$ and the leading angle of $r$ is $\theta/4$. Then if there is a robot visible to $r$ in the arc $[s-\frac{\theta}{4}R,s+\frac{\theta}{4}R)$ then $r$ changes its state to terminate and move $3\theta/4$ angular distance counterclockwise. If $r'$ is antipodal and there is no robot visible to $r$ in the arc $[s-\frac{3\theta}{4}R,s+\frac{\theta}{4}R)$, then $r$ changes the state to \texttt{off} and moves at the position of its first clockwise neighbor.

Note that, according to the definition of sure leader (confused leader, follower robot), a robot can identify itself whether it is a sure leader (confused leader, follower robot) or not.

\begin{algorithm}
 \footnotesize
The algorithm is executed by a generic robot $r$ with initial state \texttt{off}\;
\KwIn{The set of points occupied by robots visible to $r$} 
 \KwOut{Destination point for robot $r$}
 \eIf{there is a robot visible}
 {
    \uIf{there is no multiplicity point}
    {
        \uIf{if the robot $r$ has state \texttt{off}}
        {
            \uIf{the robot $r$ is the sure leader}
            {
                move to it's clockwise neighbour's position\;
            }
            \ElseIf{the robot $r$ is a confused leader}
            {
                \uIf{the clockwise first neighbor of $r$ is safe}
                {
                    move to the position of its clockwise first neighbor
                } 
                \ElseIf{$C_0(r)$ configuration does not have another confused leader other than $r$}
                {
                    let leading angle of $r$ be $\theta$\;
                    change the state to \texttt{moveHalf} and move $\frac{\theta}{2}$ angular distance clockwise\; 

                }
            }
        }
        
        \uElseIf{the robot $r$ has the state \texttt{moveHalf}}        
        {
            let the leading angle of $r$ be $\frac{\theta}{2}$\;
            let $s$ be the antipodal position of $r$\;
            \uIf{clockwise neighbor of $r$ is non antipodal}
            {
                change the state to \texttt{terminate}\;
            }
            \uElseIf{there is no robot in the arc $[s-\frac{\theta}{2}R,s+\frac{\theta}{2}R)$}
            {
                change the state to \texttt{moveMore}\;
                move $\frac{\theta}{4}$ angular distance clockwise\;  
            }
            \ElseIf{there is a robot in the arc $[s-\frac{\theta}{2}R,s+\frac{\theta}{2}R)$}
            {
                change the state to \texttt{terminate}\;
                move $\frac{\theta}{2}$ angular distance counterclockwise\;       
            }
            
        }
        \ElseIf{the robot $r$ has the state \texttt{moveMore}}  
        {
            let the leading angle of $r$ be $\frac{\theta}{4}$\;
            let $s$ be the antipodal position of $r$\;
            \uIf{clockwise neighbor of $r$ is non antipodal}
            {
                change the state to \texttt{terminate}\;
            }
            \uElseIf{there is no robot in the arc $[s-\frac{3\theta}{4}R,s+\frac{\theta}{4}R)$}
            {
                change the state to \texttt{off}\;
                move to the position of its clockwise first neighbor\; 
            }
            \ElseIf{there is a robot in the arc $[s-\frac{3\theta}{4}R,s+\frac{\theta}{4}R)$}
            {
                change the state to \texttt{terminate}\;
                move $\frac{3\theta}{4}$ angular distance counterclockwise\;       
            }
            
        }
                   
    }
    \uElseIf{there is a multiplicity point but the robot $r$ is not at any multiplicity point}
    {
        \If{its clockwise or counter-clockwise neighbor is a multiplicity point}
        {
            Move to the closer multiplicity point
        }
    }
    \ElseIf{Only visible position is a multiplicity point}
    {
        \If{clockwise angular distance from the multiplicity point is $<\pi/2$}
        {
            Move to the multiplicity point
        }
    }
    
 }
 {
    Move $\frac{\pi}{2}$ distance in clockwise direction
 }
\caption{Gathering algorithm for visibility $\pi$ }
\label{algo:gathering}
\end{algorithm}

Before formally presenting the proposed algorithm, we give the definition of a \textit{safe} clockwise neighbor.

\begin{definition}[Safe neighbor]
    Suppose $r$ is a confused expected leader and $s$ is the first clockwise neighbor of $r$. The robot $s$ is said to be a safe neighbor of $r$ if the first clockwise neighbor of the true leader of $C_1(r)$ configuration is not antipodal to $s$.
\end{definition}

 
 Next, we formally present the algorithm given in Algorithm~\ref{algo:gathering}.
 
 Next, we categorise an initial configuration, rotationally asymmetric configuration with no multiplicity points in the following:
 \begin{itemize}
 \item \textit{Configuration-A:} Only the expected leader is the true leader of the configuration. If the expected leader is a confused leader then its clockwise first neighbor is safe.
 \item \textit{Configuration-B:} There are two expected leaders in the configuration.
 \begin{itemize}
     \item \textit{Configuration-BI:} when the confused leader which is not the true leader finds its clockwise first neighbor safe.
     \item \textit{Configuration-BII:} when the confused leader which is not the true leader finds its clockwise first neighbor unsafe.
 \end{itemize}
 \item \textit{Configuration-C:} One expected leader which is a confused leader sees that its clockwise first neighbor is not safe.
 \end{itemize}

 \section{Correctness}
We will prove that if all robots are initially placed on a rotationally asymmetric configuration with no multiplicity point then on finite time execution of Algorithm~\ref{algo:gathering} gathering of all robots eventually will occur at a point and no longer move in the asynchronous scheduler. First, we show that from the initial configuration after finite execution of Algorithm~\ref{algo:gathering} at least one and at most two multiplicity points will be created by robots. Then from a configuration with one or two multiplicity point the robots eventually gather at one of the multiplicity points and do not move further.

  \begin{lemma}\label{onesure}
  If the initial configuration is type configuration-A, then after finite time execution of Algorithm~\ref{algo:gathering} at least one multiplicity point will form.
 \end{lemma}

  \begin{lemma}\label{oneslonecl}
 If the initial configuration is type configuration-B, then after finite time execution of Algorithm~\ref{algo:gathering} at least one multiplicity point will form.
 \end{lemma}

\begin{lemma}\label{onecl}
 If the initial configuration is type configuration-C, then after finite time execution of Algorithm~\ref{algo:gathering} at least one multiplicity point will form.
 \end{lemma}


    \begin{figure}[ht] 
     \centering 
     \includegraphics[width=1\linewidth]{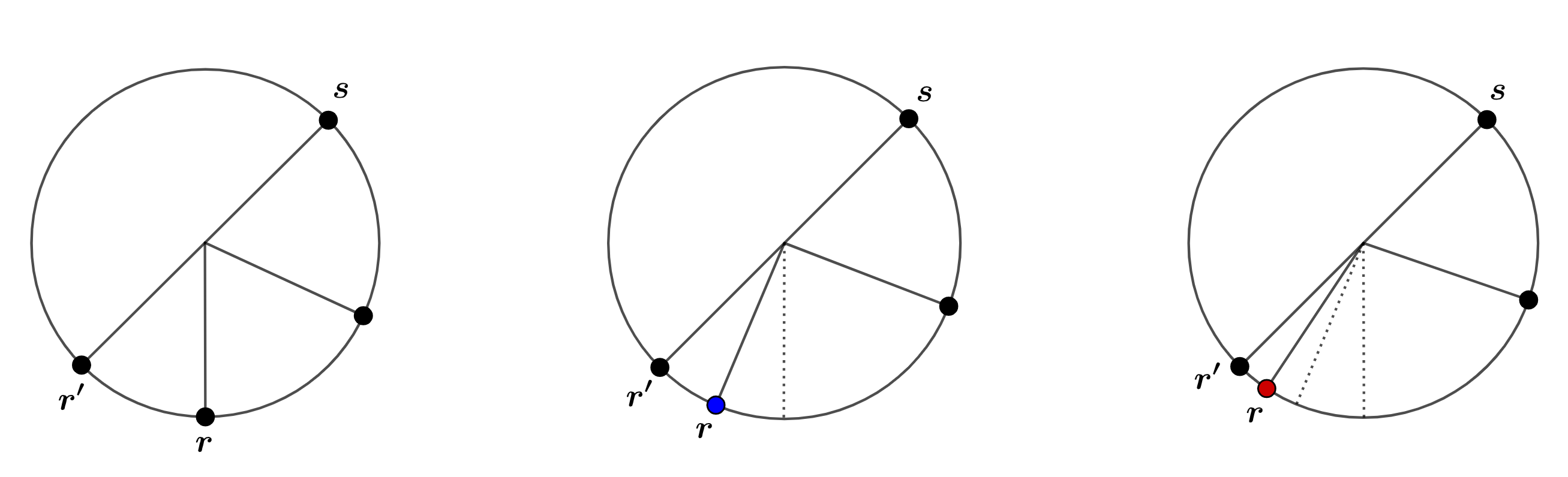}
     \caption{Movement of one confused leader ($r$) in type configuration-C where initial color of $r$ is \texttt{off}, blue color indicates \texttt{moveHalf} and red color indicates \texttt{moveMore}}
     \label{Fig:c}
    \end{figure}

 \begin{theorem}\label{theo}
Let $\mathcal{C}$ be a rotationally asymmetric configuration with no multiplicity point, then after finite execution of Algorithm~\ref{algo:gathering} by the robots on $\mathcal{C}$ at least one multiplicity point will be created.
 \end{theorem}
  \begin{proof}
 Let $\mathcal{C}$ be a rotationally asymmetric configuration with no multiplicity point. Then there can be three exhaustive possible configurations, Configuration-A, Configuration-B and Configuration-C. 
 
 From lemma~\ref{onesure}, lemma~\ref{oneslonecl}, and lemma~\ref{onecl} we can say that in any type of configuration at least one multiplicity point will form.
  \end{proof}

  \begin{lemma}\label{twomult}
 From any rotationally asymmetric configuration with no multiplicity point, by finite time execution of Algorithm~\ref{algo:gathering} the robots can form at most two multiplicity points and then all robots gather at a point on the circle.
 \end{lemma}

 Hence, we can conclude the following theorem.
  \begin{theorem}
 There exists a gathering algorithm that gathers any set of robots with finite memory and $\pi$ visibility from any initial rotationally asymmetric configuration under asynchronous scheduler.   
 \end{theorem}
\section{Conclusion}
In this paper, we present a gathering algorithm of robots with finite memory on a circle under an asynchronous scheduler with visibility $\pi$. Robots are initially at distinct positions on the circle forming any rotationally asymmetric configuration. We assume that each robot has finite persistent memory. For future studies on this problem, it will be interesting, if one can give a gathering algorithm when robots are oblivious or the visibility is less than $\pi$.

%

%
%
 \bibliographystyle{splncs04}
 \bibliography{samplepaper}

\begin{thebibliography}{10}
\providecommand{\url}[1]{\texttt{#1}}
\providecommand{\urlprefix}{URL }
\providecommand{\doi}[1]{https://doi.org/#1}

\bibitem{'agmon-06'}
Agmon, N., Peleg, D.: Fault-tolerant gathering algorithms for autonomous mobile
  robots. {SIAM} J. Comput.  \textbf{36}(1),  56--82 (2006)

\bibitem{ando99}
Ando, H., Oasa, Y., Suzuki, I., Yamashita, M.: Distributed memoryless point
  convergence algorithm for mobile robots with limited visibility. IEEE
  Transactions on Robotics and Automation  \textbf{15}(5),  818--828 (1999).
  \doi{10.1109/70.795787}

\bibitem{BhagatCDM22}
Bhagat, S., Chakraborty, A., Das, B., Mukhopadhyaya, K.: Gathering over meeting
  nodes in infinite grid\({}^{\mbox{*}}\). Fundam. Informaticae
  \textbf{187}(1),  1--30 (2022). \doi{10.3233/FI-222128}

\bibitem{BhagatMM19}
Bhagat, S., Mukhopadhyaya, K., Mukhopadhyaya, S.: Computation under restricted
  visibility. In: Flocchini, P., Prencipe, G., Santoro, N. (eds.) Distributed
  Computing by Mobile Entities, Current Research in Moving and Computing,
  Lecture Notes in Computer Science, vol. 11340, pp. 134--183. Springer (2019).
  \doi{10.1007/978-3-030-11072-7\_7}

\bibitem{BouchardDD15}
Bouchard, S., Dieudonn{\'{e}}, Y., Ducourthial, B.: Byzantine gathering in
  networks. CoRR  \textbf{abs/1504.01623} (2015)

\bibitem{BFKSW21}
Buchin, K., Flocchini, P., Kostitsyna, I., Peters, T., Santoro, N., Wada, K.:
  Autonomous mobile robots: Refining the computational landscape. In: {IEEE}
  International Parallel and Distributed Processing Symposium Workshops,
  {IPDPS} Workshops 2021, Portland, OR, USA, June 17-21, 2021. pp. 576--585.
  {IEEE} (2021). \doi{10.1109/IPDPSW52791.2021.00091}

\bibitem{Casteigts10}
Casteigts, A., Flocchini, P., Quattrociocchi, W., Santoro, N.: Time-varying
  graphs and dynamic networks. CoRR  \textbf{abs/1012.0009} (2010)

\bibitem{CieliebakFPS12}
Cieliebak, M., Flocchini, P., Prencipe, G., Santoro, N.: Distributed computing
  by mobile robots: Gathering. {SIAM} J. Comput.  \textbf{41}(4),  829--879
  (2012)

\bibitem{'CzyzowiczKP12'}
Czyzowicz, J., Kosowski, A., Pelc, A.: How to meet when you forget: log-space
  rendezvous in arbitrary graphs. Distributed Comput.  \textbf{25}(2),
  165--178 (2012)

\bibitem{jczyzowicz2013}
Czyzowicz, J., Kosowski, A., Pelc, A.: Deterministic rendezvous of asynchronous
  bounded-memory agents in polygonal terrains. THEORY OF COMPUTING SYSTEMS
  \textbf{52}(2),  179--199 (2013)

\bibitem{DAngeloSKN16}
D'Angelo, G., Stefano, G.D., Klasing, R., Navarra, A.: Gathering of robots on
  anonymous grids and trees without multiplicity detection. Theor. Comput. Sci.
   \textbf{610},  158--168 (2016). \doi{10.1016/j.tcs.2014.06.045}

\bibitem{DefagoP0MPP16}
D{\'{e}}fago, X., Potop{-}Butucaru, M.G., Cl{\'{e}}ment, J., Messika, S.,
  Parv{\'{e}}dy, P.R.: Fault and byzantine tolerant self-stabilizing mobile
  robots gathering - feasibility study -. CoRR  \textbf{abs/1602.05546} (2016)

\bibitem{degener11}
Degener, B., Kempkes, B., Langner, T., Meyer auf~der Heide, F., Pietrzyk, P.,
  Wattenhofer, R.: A tight runtime bound for synchronous gathering of
  autonomous robots with limited visibility. In: Proceedings of the
  Twenty-Third Annual ACM Symposium on Parallelism in Algorithms and
  Architectures. p. 139–148. SPAA '11, Association for Computing Machinery,
  New York, NY, USA (2011). \doi{10.1145/1989493.1989515}

\bibitem{Dieudonn-09}
Dieudonn{\'{e}}, Y., Petit, F., Villain, V.: Leader election problem versus
  pattern formation problem. CoRR  \textbf{abs/0902.2851} (2009)

\bibitem{FeinermanKKR14}
Feinerman, O., Korman, A., Kutten, S., Rodeh, Y.: Fast rendezvous on a cycle by
  agents with different speeds. In: Distributed Computing and Networking - 15th
  International Conference, {ICDCN} 2014, Coimbatore, India, January 4-7, 2014.
  Proceedings. pp. 1--13 (2014). \doi{10.1007/978-3-642-45249-9\_1}

\bibitem{'FlocchiniKKSY19'}
Flocchini, P., Killick, R., Kranakis, E., Santoro, N., Yamashita, M.: Gathering
  and election by mobile robots in a continuous cycle. In: Lu, P., Zhang, G.
  (eds.) 30th International Symposium on Algorithms and Computation, {ISAAC}
  2019, December 8-11, 2019, Shanghai University of Finance and Economics,
  Shanghai, China. LIPIcs, vol.~149, pp. 8:1--8:19. Schloss Dagstuhl -
  Leibniz-Zentrum f{\"{u}}r Informatik (2019).
  \doi{10.4230/LIPIcs.ISAAC.2019.8}

\bibitem{flochhini19}
Flocchini, P., Prencipe, G., Santoro, N. (eds.): Distributed Computing by
  Mobile Entities, Current Research in Moving and Computing, Lecture Notes in
  Computer Science, vol. 11340. Springer (2019)

\bibitem{Flocchini05}
Flocchini, P., Prencipe, G., Santoro, N., Widmayer, P.: Gathering of
  asynchronous robots with limited visibility. Theor. Comput. Sci.
  \textbf{337}(1-3),  147--168 (2005). \doi{10.1016/j.tcs.2005.01.001}

\bibitem{GoswamiSGS22}
Goswami, P., Sharma, A., Ghosh, S., Sau, B.: Time optimal gathering of myopic
  robots on an infinite triangular grid. In: Devismes, S., Petit, F., Altisen,
  K., Luna, G.A.D., Anta, A.F. (eds.) Stabilization, Safety, and Security of
  Distributed Systems - 24th International Symposium, {SSS} 2022,
  Clermont-Ferrand, France, November 15-17, 2022, Proceedings. Lecture Notes in
  Computer Science, vol. 13751, pp. 270--284. Springer (2022).
  \doi{10.1007/978-3-031-21017-4\_18}

\bibitem{GuilbaultP11}
Guilbault, S., Pelc, A.: Gathering asynchronous oblivious agents with local
  vision in regular bipartite graphs. In: Kosowski, A., Yamashita, M. (eds.)
  Structural Information and Communication Complexity - 18th International
  Colloquium, {SIROCCO} 2011, Gdansk, Poland, June 26-29, 2011. Proceedings.
  Lecture Notes in Computer Science, vol.~6796, pp. 162--173. Springer (2011).
  \doi{10.1007/978-3-642-22212-2\_15}

\bibitem{HuusK15}
Huus, E., Kranakis, E.: Rendezvous of many agents with different speeds in a
  cycle. In: Papavassiliou, S., Ruehrup, S. (eds.) Ad-hoc, Mobile, and Wireless
  Networks - 14th International Conference, {ADHOC-NOW} 2015, Athens, Greece,
  June 29 - July 1, 2015, Proceedings. Lecture Notes in Computer Science,
  vol.~9143, pp. 195--209. Springer (2015). \doi{10.1007/978-3-319-19662-6\_14}

\bibitem{LunaFPPSV17}
Luna, G.A.D., Flocchini, P., Pagli, L., Prencipe, G., Santoro, N., Viglietta,
  G.: Gathering in dynamic rings. CoRR  \textbf{abs/1704.02427} (2017)

\bibitem{LunaFSVY20a}
Luna, G.A.D., Flocchini, P., Santoro, N., Viglietta, G., Yamashita, M.: Meeting
  in a polygon by anonymous oblivious robots. Distributed Comput.
  \textbf{33}(5),  445--469 (2020)

\bibitem{'LunaUVY20'}
Luna, G.A.D., Uehara, R., Viglietta, G., Yamauchi, Y.: Gathering on a circle
  with limited visibility by anonymous oblivious robots. In: Attiya, H. (ed.)
  34th International Symposium on Distributed Computing, {DISC} 2020, October
  12-16, 2020, Virtual Conference. LIPIcs, vol.~179, pp. 12:1--12:17. Schloss
  Dagstuhl - Leibniz-Zentrum f{\"{u}}r Informatik (2020).
  \doi{10.4230/LIPIcs.DISC.2020.12}

\bibitem{PoudelS17}
Poudel, P., Sharma, G.: Universally optimal gathering under limited visibility.
  In: Spirakis, P.G., Tsigas, P. (eds.) Stabilization, Safety, and Security of
  Distributed Systems - 19th International Symposium, {SSS} 2017, Boston, MA,
  USA, November 5-8, 2017, Proceedings. Lecture Notes in Computer Science, vol.
  10616, pp. 323--340. Springer (2017). \doi{10.1007/978-3-319-69084-1\_23}

\bibitem{Souissi06}
Souissi, S., D{\'{e}}fago, X., Yamashita, M.: Using eventually consistent
  compasses to gather oblivious mobile robots with limited visibility. In:
  Datta, A.K., Gradinariu, M. (eds.) Stabilization, Safety, and Security of
  Distributed Systems, 8th International Symposium, {SSS} 2006, Dallas, TX,
  USA, November 17-19, 2006, Proceedings. Lecture Notes in Computer Science,
  vol.~4280, pp. 484--500. Springer (2006). \doi{10.1007/978-3-540-49823-0\_34}

\end{thebibliography}

\end{document}